\documentclass[11pt]{amsart}

\usepackage{amsmath,amsfonts}
\usepackage{algorithm,algpseudocode}
\usepackage{array}
\usepackage[caption=false,font=normalsize,labelfont=sf,textfont=sf]{subfig}
\usepackage{textcomp}
\usepackage{url}
\usepackage{verbatim}
\usepackage{graphicx}
\usepackage{cite}
\usepackage{amssymb,amsmath,amsthm,latexsym,amsfonts}
\usepackage{amsmath, amssymb, amsthm}
\usepackage[mathscr]{eucal}
\usepackage{amscd}
\usepackage{epsfig,fancyhdr,enumerate}
\usepackage{mathtools}
\usepackage{mathrsfs}
\usepackage{multirow}
\usepackage{multicol}
\usepackage{linguex} 
\usepackage{xcolor}
\usepackage{color}
\usepackage{float}

\usepackage[T1]{fontenc}
\usepackage{graphicx}
\allowdisplaybreaks

\usepackage{soul,times,blkarray,comment}
\usepackage{pdflscape} 
\usepackage{longtable}
\usepackage[utf8]{inputenc}
\usepackage{ulem}
\usepackage{cancel} 
\usepackage{url}

\theoremstyle{definition}
\newtheorem{theorem}[algorithm]{Theorem}

\newtheorem{lemma}[algorithm]{Lemma}

\newtheorem{alg}[algorithm]{Algorithm}
\newtheorem{example}[algorithm]{Example}
\newtheorem{definition}[algorithm]{Definition}
\newtheorem{remark}[algorithm]{Remark}

\def\x{{\mathbf x}}
\def\y{{\mathbf y}}
\def\cc{{\mathbf c}}
\def\dd{{\mathbf d}}
\def\mm{{\mathbf m}}
\def\u{{\mathbf u}}
\def\v{{\mathbf v}}

\def\0{{\mathbf 0}}
\def\1{{\mathbf 1}}

\def\M{{\mathcal M}}
\def\F{{\mathbb F}}

\def\A{{\mathscr A}}

\def\CC{{\mathcal C}}

\def\DD{{\mathcal D}}

\newenvironment{smatrix2}{\left(\begin{smallmatrix}}{\end{smallmatrix}\right)}

\begin{document}

\title[Algorithms of self-synchronizing single-deletion-correcting codes]
{Algorithms of self-synchronizing single-deletion-correcting codes}

\author{
	Whan-Hyuk Choi}
\address{Department of Mathematics
	Kangwon National University
	Chuncheon 24341, Korea}
\address{Kangwon Research Institute of Mathematical Sciences\\
	Kangwon National University\\
	Chuncheon 24341, Korea\\  
	e-mail: whchoi@kangwon.ac.kr\\
}


\maketitle


\begin{abstract}
		This study explores the self-synchronization problem in DNA coding, specifically addressing single-deletion errors without using delimiters between codewords. We aim to identify the beginning of each codeword without using delimiters, enhancing the transmission efficiency. The motivation arises from the inefficiency of adding meaningless symbols as delimiters, decreasing the information rate. In addition, the historical context in biology, specifically Francis Crick's proposal of ``codes without commas'' for DNA sequences, inspires this investigation. We introduce a novel approach for correcting single-deletion errors in continuous transmissions without delimiters, distinguishing the beginning and end of each codeword. This approach is based on the properties of {\it complementary information set codes}, which is used to present an algorithm for {\it single-deletion correcting codes} with self-synchronizing capability. Accordingly, we present encoding and decoding algorithms for self-synchronizing single-deletion correcting DNA codes with concrete examples.
	\end{abstract}
	
	{\bf Keywords:} 
		Single-deletion-correcting code, DNA codes, reversible self-dual code, self-synchronizing block codes, and non-binary code.

	\maketitle

\section{Introduction}\label{sec:introduction}
Since the inception of the coding theory, error-correcting codes have focused on substitution errors that involve symbol changes \cite{HuffmanPless, MS}. However, practical communication scenarios often involve a spectrum of errors beyond symbol substitutions, encompassing deletions, insertions, and erasures, collectively known as {\it synchronization errors} \cite{levenshtein1966binary,Sal}. In our previous work \cite{choi2020, kim2021}, we studied deletion-error-correcting codes and their application to DNA codes. Moreover, we introduce a construction algorithm for {\it single-single-deletion-correcting (SDC) DNA codes} \cite{choi2020}.
	When transmitting encoded data, it is assumed that the sender and receiver know the length as well as the beginning and end of each codeword. The start and end of each codeword can be distinguished by the proper placement of {\it delimiters}; for example, a fixed number of zeroes are placed between each codeword. Under this assumption, we proposed an algorithm for correcting single insertion/deletion errors in \cite{choi2020}. 
	
	However, adding meaningless symbols, such as commas, to distinguish the beginning and end of a codeword reduces transmission efficiency as the overall length of the information increases and the information rate of the code decreases. Thus, we aim to answer the following question: How can we identify the beginning of each codeword without delimiters placed between the codewords, even if they suffer from deletion errors?  We call this the {\it self-synchronization problem}. 
	
	For instance, suppose that one receives a sequence of DNA codes of length 8, separated by a delimiter, symboled by |, as follows:
$$
		GATCCTAG |AGTTACT| GGATTGTC|GGGTTGG| CGTCCTGC.
$$
	Because the delimiter’s positions are known, the receiver can tell that this sequence has 5 codewords and that a single deletion error has occurred in the second and fourth codewords because they have only seven symbols. This type of error can be corrected using the algorithm in \cite{kim2021}.
	However, with no delimiter for each codeword, the receiver receives a sequence of 38 consecutive symbols, as follows:
$$GATCCTAG AGTTACT GGATTGTC GGGTTGG CGTCCTGC.$$
	In this case, at least two deletion errors can be noticed when checking the sequence's length. However, it is impossible to determine the positions of the deletion errors because the receiver does not know where each codeword begins or ends.
	
	This study introduces a novel approach for correcting a single-deletion error, even when transmitting codewords continuously without delimiters, by distinguishing the beginning and end of each codeword. Discarding delimiters is the primary distinction between this and earlier studies on synchronization errors.

	Interestingly, the first self-synchronization problem emerged in biology shortly after the discovery of deoxyribonucleic acid (DNA) by one of its discoverers, Francis Crick \cite{crick}. Crick aimed to solve a mathematical issue concerning DNA sequences in connection with protein synthesis. In \cite{crick}, Crick proposed \textit{codes without commas} (equivalently, codes without delimiters), which refer to DNA codes composed of three DNA symbols encoding each amino acid; this solved a sort of synchronization problem of DNA sequences. Although Crick's solution was proved biologically wrong, his idea was developed by mathematicians interested in the synchronization problem \cite{Golomb, levy, Massey1972}.
	
	In this study, we investigated the self-synchronization problem of DNA codes. We introduced encoding and decoding algorithms for SDC DNA codes with self-synchronizing capabilities. These algorithms detect the location of a single deletion error in each DNA codeword, correct the error, and distinguish the beginning and end of each DNA codeword. Even though there has been extensive research on deletion error correction, especially in the context of DNA sequences, almost of them has assumed the presence of delimiters between codewords\cite{Yan2023, Hanna2023, Haeupler2021}. Therefore, as far as we know, this is the first study to provide explicit algorithms for self-synchronizing SDC DNA codes without using delimiters.
	
	The remainder of this paper is organized as follows. We begin with the preliminaries in Section 2. In addition, we summarize crucial results from \cite{choi2020} and \cite{kim2021} and discuss the properties of DNA codes.
	Theorem \ref{main_thm3}, Algorithm \ref{alg:encoding1}, and Algorithm \ref{alg:decoding1} are presented in Section 3. Section 4 describes the implementation of the novel algorithm.
	

	\section{Preliminaries}
	
	\subsection{Single-deletion-correcting codes}
	
	Let $\F_q$ be the finite field of order $q$ for a prime power $q$. A subset $\CC$ of $\F_q^n$ is called a {\it code} of length $n$ over $\F_q$. In particular, if $\CC$ is a $k$-dimensional subspace of $\F_q^n$, $\CC$ is called a {\it $q$-ary linear code} of length $n$ and dimension $k$, which we denote as {\it $[n,k]_q$ code}. Each element of a code is called a {\it codeword }.
	A code $\CC$ of length $n$ is called {\it systematic} if there exists a subset $I$ of $\{1,2,\ldots,n\}$ (called an {\it information set} of $\CC$) such that every possible tuple of length $\mid I\mid$ occurs in exactly one codeword in $\CC$ within the specified coordinates $x_i$; $i \in I$ \cite{CGKS, kim2016cis}.
	If $\CC$ is a systematic code with information set $I$ of size $k$, then there exists a one-to-one correspondence between $\F_q^k$ and distinct $q^k$ codewords in $\CC$ whose coordinates are contained in $I$. If the first $k$ coordinates form the information set, the code has a unique generator matrix of the form $(E_k \mid  A)$, where $E_k$ is a $k \times k$ identity matrix and $A$ is a $k \times n-k$ matrix. Such a generator matrix is said to be in {\it standard form}. 

A {\it complementary information set (CIS) code} is a special type of systematic code: a CIS code \(\CC\) of length $n$ over $\F_q$ is a \([2n, n]_q\) code which has two disjoint information sets \(I\) and \(J\), each of size \(n\). In other words, every vector in \(\F_q^n\) appears exactly once in the coordinates of \(\CC\) restricted to \(I\), and also exactly once in the coordinates restricted to \(J\). 	For details of the CIS codes, refer \cite{CGKS}.	
	
The following lemma characterizes CIS codes over $\F_q$.

\begin{lemma}{\rm{\cite[Lemma 4.1]{CGKS}}}\label{lem:systematic}
	If a $[2n,n]$ code $\CC$ over $\F_q$ has generator matrix $(I\mid A)$ with $A$
	invertible, then $\CC$ is a CIS code with the systematic partition.
	Conversely, every CIS code is equivalent to a code with generator matrix in that form.
\end{lemma}

	\begin{example}\label{ex1}
Consider a \([4, 2]_3\) linear code $\CC_1$ having the generator matrix in standard form 
$$\begin{pmatrix} 1&0&1&1\\0&1&1&0 \end{pmatrix},$$ which means 
$$\CC_1=\{0000, 1011, 0110, 1121, 2022, 0220, 2212, 1201, 2102 \}.$$
It is easy to check that all nine possible vectors in $\F_3^2$ appear in the first two coordinates, as well as in the last two coordinates.
Since both the first and last two coordinates form disjoint information sets of size two, $\CC_1$ is a CIS code.
\end{example}

	Let $\x$ be a codeword in code $\CC$ of length $n$. If a vector $\y \in \F_q^{n-1}$ is obtained from $\x$ by deleting one symbol of $\x$, then $\y$ is called a {\it subword} of $\x$. 
	We denote a subword $\y$ of $\x$ as $\y=\x^i$ if the $i$th symbol of $\x$ is deleted.
	Let $D_1(\x)$ denote a set of subwords in $\x$. A code $\CC$ is said to be an {\it SDC code} if $D_1(\x_1) \cap  D_1(\x_2) = \emptyset$ for all $\x_1 , \x_2 \in \CC$, $\x_1 \ne \x_2$.

	Let $\u$ and $\v$ be two codewords in code $\CC$. The {\it Levenshtein distance} $d_l(\u,\v)$ between $\u$ and $\v$ is defined as the smallest number of insertions and deletions required to transform $\u$ into $\v$: Levenshtein distance is a metric. The {\it minimum Levenshtein distance} of $\CC$, denoted by $d_l(\CC)$, is the smallest Levenshtein distance between distinct codewords in $\CC$.
	The {\it Hamming distance} $d_h(\u,\v)$ between $\u$ and $\v$ is defined as the number of coordinates in which $\x$ and $\y$ differ. The {\it minimum Hamming distance} of $\CC$, denoted by $d_h(\CC)$, is the smallest Hamming distance between the distinct codewords in $\CC$. Code $\CC$ can correct $t$ substitution errors only if $d_h(\CC) > 2t$. 
	Similarly, code $\CC$ can correct $e$ deletion/insertion errors if and only if $d_l(\CC) > 2e$. 
	
	Next, we introduce some results of our previous study \cite{choi2020} without proof. The following theorem and remark are the main results from \cite{choi2020}, which show that an SDC code can be created from a CIS code by inserting two identical symbols with a specific rule. 
	
	\begin{theorem}{\rm{\cite[Theorem 3.5]{choi2020}}}\label{CIStoSDCbin}
		Let $\CC$ be a CIS code of length $2n$ over $\F_q$ and let $$\phi : \F_q^{2n} \to \F_q$$ be a map  defined by
		$$\phi(\x) = x_n+1,$$ where $\x = (x_i) \in \F_q^{2n}$.
		We also define a vector $$\x_{\phi} = (x_1, \cdots, x_n,\phi(\x) ,\phi(\x) ,x_{n+1} , \cdots ,x_{2n}) \in \F_q^{2n+2}$$ obtained by adding two $\phi(\x)$'s between the $n$th and $(n+1)$th positions of $\x$ for every codeword $\x$ in $\CC$. Then the set of vectors $\x_{\phi}$ for all codewords $\x$ in $\CC$, that is, $$\CC_{\phi} = \{\x_{\phi} \mid \x \in \CC \},$$ is an SDC code.	
	\end{theorem}

	\begin{example}
		Consider the code $$\CC_1=\{0000, 1011, 0110, 1121, 2022, 0220, 2212, 1201, 2102 \}$$ in Example \ref{ex1}, which consists of nine codewords.
		The minimum Levenshtein distance $d_l(\CC_1)$ is 2 since Levenshtein distance between codewords $1011$ and $0110$ is 2. 
		By applying the map $\phi$ of Theorem \label{CIStoSDCbin} on $\CC_1$, we obtain 
			\begin{multline*}	$${\CC_1}_\phi=\{001100, 101111, 012210, 112221,\\ 201122, 020020, 220012, 120001, 212202 \}.$$
		\end{multline*}
		It is routine to check that $d_l({\CC_1}_\phi)=4$ and ${\CC_1}_\phi$ is an SDC code.
	\end{example}

	\begin{remark}\label{remark}
	
		If $\phi$ is defined by $\phi(\x) = x_n+a$ for any fixed nonzero element $a$ in $\F_q$, then set $\CC_{\phi}$ is an SDC code. If we define $\phi$ such that the image $\phi(\x)$ is different from $x_n$ for every codeword $\x$ in CIS code $\CC$, then we can obtain SDC code $\CC_{\phi}$.
		
	\end{remark}
	
	As Remark \ref{remark} points out, the key to constructing the SDC code from CIS codes of length $2n$ involves choosing the symbol $\phi(\x)$ that is different from the $n$-th symbol $x_n$ of each codeword $\x$ of the CIS code and inserting it twice in the middle of the codeword. For proof and details, refer \cite{choi2020}. 
	
	The following is from Algorithm 3.7 in \cite{choi2020}. The algorithm focuses on single-deletion-error correction in a received vector with delimiters from a codeword in $\CC_{\phi}$ where $\CC$ is a CIS code of length $2n$.
	
	\begin{algorithm}[H]
		\caption{Decoding algorithm for single-deletion-error correction in \cite{choi2020}}\label{alg:prev_decoding}
		\begin{algorithmic}[1]
			\Require a received vector $\x$ through a single-deletion channel from $\CC_{\phi}$ over $\F_q$		
			\Ensure the codeword $\cc$ in $\CC$ decoded from $\x$
			\State $2n \gets$ the length of $\CC$ ; $L \gets$ the length of $\x$
			\If{$L = 2n+1$}  
			\State decompose $\x= \u \oplus (x_{n+1}) \oplus \v$, where $\u, \v \in \F_q^n$.
			\State $\phi \gets x_{n+1}$
			\State $u_n \gets $the last symbol of $\u$
			\If {$u_n$ = $x_{n+1}$}
			\State there is no deletion in $\v$
			\State take $\y=\v$ with the set $I=\{n+1,\ldots,2n\}$.
			\ElsIf {$u_n$ $\ne x_{n+1}$}
			\State there is no deletion in $\u$
			\State take $\y=\u$ with the set $I=\{1,\ldots,n\}$.
			\EndIf
			\State obtain $\cc$ generated by $\y$ with the information set $I$.
			\ElsIf{$L = 2n+2$}  
			\State decompose $\x=\u \oplus (x_{n+1}, x_{n+2}) \oplus \v$, where $\u, \v \in \F_q^n$.
			\State $\phi \gets x_{n+1}$
			\State $\cc \gets \u \oplus \v$
			\EndIf
			\State return $\cc, \phi$ 
		\end{algorithmic}
	\end{algorithm}

	Compared to \cite{choi2020}, this study focuses on the self-synchronization problem: How can we identify the beginning of each codeword when there are no delimiters between codewords and the transmitted data suffer from deletion errors? Some definitions must be clarified to answer this question.
	
	\begin{definition}\label{def1}
		Let $\CC$ be a code and $(\cc_1, \cc_2, \cdots, \cc_m)$ be a sequence of $m$ codewords in $\CC$. If we delete all delimiters between codewords, we denote the sequence without the delimiter simply by $\cc_1\cc_2\cdots\cc_m$. If there exists a proper algorithm that converts the sequence $\cc_1\cc_2\cdots\cc_m$ without delimiters back to the original sequence $(\cc_1, \cc_2, \cdots, \cc_m)$, then the sequence is called a {\it self-synchronizing} sequence.
	\end{definition}

	\begin{definition}\label{def2}
		Let $\CC$ be a code	and	$\cc_1\cc_2\cdots\cc_m$ be a sequence of $m$ codewords in $\CC$ without delimiters.
		When deletion errors in $\cc_1\cc_2\cdots\cc_m$ occur {\it at most once} for each codeword but {\it not consecutively}, we call these errors {\it single-deletion-errors}. If there exists an appropriate algorithm that can correct single-deletion-errors in a sequence $\cc_1\cc_2\cdots\cc_m$ and simultaneously convert the sequence $\cc_1\cc_2\cdots\cc_m$ without delimiters back to the original sequence $(\cc_1, \cc_2, \cdots, \cc_m)$, then the sequence is called a {\it self-synchronizing single-deletion-correcting sequence}.
	\end{definition}

	\begin{definition}\label{def4}
		If every sequence made from codewords in $\CC$ with an appropriate algorithm is a {\it self-synchronizing SDC sequence}, then code $\CC$ is called a {\it self-synchronizing SDC code}.
	\end{definition}
	
	\subsection{DNA codes}
	
	{\it Deoxyribonucleic acid (DNA)} encodes genetic information of life in the DNA helix with four basic units called {\it nucleotides}: Adenine($A$), Cytosine($C$), Guanine($G$) and Thymine($T$). The DNA helix is a double strand built by joining four nucleotides and complementary base pairing, which connects {\it the Watson–Crick complement}, denoted by $A^c = T$, $T^c = A$, $C^c = G$ and $G^c = C$.
	A {\it DNA code of length $n$} is a set of tuples $(x_1,\cdots,x_n)$, where $x_i \in \{A,C,G,T\}$. A {\it DNA codeword} is an element of DNA code. We call a continuous sequence of DNA codewords a {\it DNA strand}. 
	The {\it $GC$-weight} of the DNA codeword is the number of occurrences of $C$ and $G$ in the codeword. 
	
	DNA symbols can be identified with two-digit binary numbers under the map $\delta : \F_2 \times \F_2 \rightarrow \{A,C,G,T\}$ defined by
	$$ \delta(00)=A, \delta(11)=T, \delta(10)=C ~\textrm{ and}~ \delta(01)=G.$$ Using this map $\delta$, we can encode a binary sequence of length $2n$ into a DNA sequence of length $n$. For example, a binary sequence of length 20, $10 01 10 11 01 11 00 10 01 11$, is encoded into the DNA sequence $CGCTGTACGT$. 
	
	Another method is to identify DNA symbols with four elements of $\F_4$ under the bijection $\mu : \F_4 \rightarrow \{A,C,G,T\}$ defined by
	$$ \mu(0)=A, \mu(1)=T, \mu( \omega)=C ~\textrm{ and}~ \mu( \bar{\omega})=G.$$ 
	The map $\mu$ defines the {\it complement} of the elements of $\F_4$ to be compatible with the {\it Watson-Crick complement}:
	the complement of $x$ in $\F_4$ is denoted by $x^c=x+1$. 
	We enlarged this map $\mu$ to a vector and code. Thus, a DNA code can be identified using a code over $\F_4$ under the bijection $\mu$. 
	For a vector $\x = (x_1, \cdots, x_n) \in \F_4^n$, we denote the complement of $\x$ by $\x^c = (x_1^c, \cdots, x_n^c)$ and
	{\it reverse} of $\x$ by $\x^r = (x_n, \cdots, x_1)$. The {\it reverse-complement} of $\x$ is denoted by $\x^{rc} = (x_n^c, \cdots, x_1^c)$.
	A vector $\x$ is called {\it self-reversible} ({\it self-reverse-complementary}) if $\x=\x^r$ ( $\x=\x^{rc}$).
	
	For computation, we define an injection map $\tau : \F_4 \rightarrow \M_2(\F_2)$, where $\M_2(\F_2)$ be a set of $2 \times 2$ binary matrices: $$
		\tau(0)=  \begin{smatrix2} 0&0\\0&0 \end{smatrix2}, \tau(1)= \begin{smatrix2} 1&0\\0&1 \end{smatrix2}, \tau(\omega) =  \begin{smatrix2} 0&1\\1&1 \end{smatrix2}, ~{\text and }~\tau(\bar{\omega}) = \begin{smatrix2} 1&1\\1&0 \end{smatrix2},$$
	Therefore, we can represent a DNA symbol by a matrix in $\M_2(\F_2)$ using the composite map $\tau \circ \mu^{-1}$. We denote the map $\tau \circ \mu^{-1}$ and its inverse map under restriction by $f$ and $f^{-1}$, respectively. Next, we enlarge these maps to a vector, sequence, and code. For example, 
	$$f(ACCTG) = \begin{smatrix2} 0&0&0&1&0&1&1&0&1&1\\0&0&1&1&1&1&0&1&1&0 \end{smatrix2}$$

	Regarding the properties of good DNA codes, refer \cite{gaborit2005, choi2020}. 
	Given positive integers $t$ and $e$, the following five constraints are of concern when designing a DNA code $\DD$:
	\begin{itemize}
		\item[-] Hamming distance constraint ({\it HD}): $d_h(\x,\y) \ge t$ for all $\x,\y \in \DD$ with $\x \ne \y$.
		\item[-] Reverse constraint ({\it RV}): $d_h(\x,\y^r) \ge t$ for all $\x,\y \in \DD$.
		\item[-] Reverse-complement constraint ({\it RC}): $d_h(\x,\y^{rc}) \ge t$ for all $\x,\y \in \DD$.
		\item[-] Fixed GC-content constraint ({\it GC}): GC-weights of all codewords of $\DD$ are constant.
		\item[-] Deletion/insertion constraint ({\it DI}) : $d_l(\x,\y) \ge e$ for all $\x,\y \in \DD$  with $\x \ne \y$.
	\end{itemize}

	In the following section, we provide a solution to the self-synchronization problem - how can we identify the beginning of each codeword with possible errors when no commas are separating them?  
	For the solution, we start with Theorem \ref{CIStoSDCbin} and take them one step further. Our approach involves selecting a symbol $\phi(\x)$ under certain additional conditions to construct a self-synchronizing SDC code. Thus, we can distinguish the beginning and end of each codeword in an SCD code without using commas. Furthermore, we proposed an algorithm to detect and correct one deletion error in each codeword.
	
	\section{Self-synchronizing SDC DNA codes}
	
	Hereafter, we focus on the DNA codes. In the later parts of this paper, we assume that $\CC$ is a CIS code of length $2n$ over $\F_4$ for some integer $n$. A DNA code induced from $\CC$ is denoted by $\DD=\mu(\CC)$; however, we would often identify $\CC$ and $\DD$.

	The following theorem is a variation of Theorem \ref{CIStoSDCbin} in \cite{choi2020}. We use this as a stepping-stone for constructing self-synchronizing SDC DNA codes. 
	
	\begin{theorem}\label{main_thm}
		Let $\CC$ be a CIS code of length $2n$ over $\F_q$. Let $$\psi : \F_q \times \F_q \to \F_q$$ be a map satisfying the following conditions:
		\begin{enumerate}
			\item[i)] $\psi(a,b) \ne a$ and $\psi(a,b) \ne b$ 
			\item[ii)] $\psi(a+1,b+1) \ne \psi(a,b)+1$
		\end{enumerate}
		
		We also define a vector $\x_{\psi} \in \F_q^{2n+2}$ as \begin{multline*}
			(x_1, \cdots, x_n,\psi(x_n, x_{n+1}) , \psi(x_n, x_{n+1}) ,x_{n+1} , \cdots ,x_{2n}),\end{multline*} obtained by adding symbol $\psi(x_n, x_{n+1})$ twice between the $n$th and $(n+1)$-th position of $\x$ for every codeword $\x$ in $\CC$. Then the set of vectors $\x_{\psi}$ for all codewords $\x$ in $\CC$, that is, $$\CC_{\psi} = \{\x_{\psi} \mid \x \in \CC \}$$ is an SDC code.
	\end{theorem} 
	\begin{proof}
		Based on the condition of the map $\psi$, symbol $\psi(x_n,x_{n+1})$ is always different from $x_n$ for every codeword $\x$ in $\CC$. Thus, a similar reasoning as that of the proof of Theorem \ref{CIStoSDCbin} and remark \ref{remark} proves the theorem.
	\end{proof}

	The following definition provides a map $\psi$ satisfying Theorem \ref{main_thm} when $q=4$:
	
	\begin{definition} \label{def_psi}
		We define a map $\psi$ that satisfies the conditions in Theorem \ref{main_thm} as map $\psi :  \F_4 \times \F_4 \to \F_4$ by
		$$\psi(x,y) = \begin{cases}
			x+y+\omega & \text{ if } 	x+y\in \{0,1\} \text{ and } x,y\in \{0,1\}, \\
			x+y& \text{ if } 	x+y \in \{0,1\} \text{ and } x,y \in \{\omega,\bar{\omega}\},\\
			x+y& \text{ if } 	x+y \in \{\omega,\bar{\omega}\} \text{ and } xy \ne  0, \\
			1& \text{ if } 	x+y \in \{\omega,\bar{\omega}\} \text{ and }xy= 0.\\		
		\end{cases}$$
		We also present the values of map $\psi$ in Table \ref{map_psi}.   
		
		\begin{table}[!h]
			\begin{center}
				\caption{Images of map $\psi(x,y)$}\label{map_psi}
				\begin{tabular}{|c|c|c|c|c|}
					
					\hline
					$x \textbackslash y$	&0&1&$\omega$&$\bar{\omega}$\\
					
					\hline
					0 				&  $\omega$  &$\bar{\omega}$&1 & 1 \\
					\hline
					1 				& $\bar{\omega}$ & $\omega$ & $\bar{\omega}$ & $\omega$  \\
					\hline
					$\omega$ & 1 & $\bar{\omega}$ &  0 & 1 \\
					\hline
					$\bar{\omega}$& 1 &  $\omega$ &  1 & 0\\
					\hline
				\end{tabular}

			\end{center}		
		\end{table}
		We can easily verify that $\psi(a,b) \ne a$, $\psi(a,b) \ne b$, and $\psi(a+1,b+1) \ne \psi(a,b)+1$ for each $a$ and $b$ in $\F_4$, which are the conditions in Theorem \ref{main_thm}. 	
	\end{definition}

	The following example motivates the main idea behind our approach.
	
	\begin{example}
		Consider a simple DNA code $D_1$ with four codewords of length six:
		$$D_1=\{AACCAA, TATTCG, CATTGT, GAGGTC \}.$$ 
		This DNA code $D_1$ is a subcode of $\mu({\CC_\psi})$ for a CIS code $\CC$ of length four over $\F_4$ with generator matrix $\begin{smatrix2} 1&0& \omega & \bar{\omega} \\ 0&1& \bar{\omega} & {\omega} \end{smatrix2}$ and the map $\psi$ is from Definition \ref{def_psi}. 
		Thus, it is easy to verify that $D_1$ is an SDC code.
		Suppose some data are encoded in a DNA strand made from $D_1$. 
		$$ C{\color{red}A}TTGT,GAG{\color{red}G}TC,CATTG{\color{red}T},TAT{\color{red}T}CG , A{\color{red}A}CCAA.$$
		Each codeword is separated by commas that act as delimiters. Now, assume that there are no commas and each codeword allows a single-deletion error so that the red symbols are deleted. Then, the DNA strand becomes 
		$$ CTTGTGAGTCCATTGTATCGACCAA.$$
		How can we decode this DNA strand to recover its original form? Firstly, we examine the first six symbols $CTTGTG$, and decode this to the codeword $C{\color{red}A}TTGT$ using Algorithm \ref{alg:prev_decoding}: $$ C{\color{red}A}TTGT,GAGTCCATTGTATCGACCAA.$$
		Then, we notice that symbol $G$, the last symbol of $CTTGTG$, is the first symbol of the second possible codeword $GAGTCC$. Next, $GAGTCC$ is decoded to a codeword $GAGGTC$ in $D_1$ using Algorithm \ref{alg:prev_decoding}:$$ C{\color{red}A}TTGT,GAG{\color{red}G}TC,CATTGTATCGACCAA.$$
		Regarding the third codeword $CATTGT$, there is a problem: we cannot notice the deletion of the last symbol $T$ of $CATTG{\color{red}T}$ because the fourth codeword also begins with the symbol $T$. If we admit the third codeword to be $CATTGT$, the corrected DNA strand becomes 
		$$ C{\color{red}A}TTGT,GAG{\color{red}G}TC,CATTG{\color{blue}T}, ATCGACCAA.$$
		Consequently, because the remaining strand $ATCGACCAA$ is considered to have three deletions in two codewords, we cannot decode the remaining sequences correctly.
	\end{example}

	From the previous example, we make the following observations.
	\begin{itemize}
		\item[1)]Observation 1: If it is possible to decode the DNA strand codeword by codeword, from the first to the last in turn, then the whole DNA strand may be decoded. 
		\item[2)]Observation 2: If the last symbol of the previous codeword is identical to the first symbol of the current codeword, there may be confusion about the beginning of the codeword.
	\end{itemize}
	
	Motivated by these observations, we propose three novel algorithms, Algorithms \ref{alg:novel_decoding}, \ref{alg:encoding1} and \ref{alg:decoding1}, for encoding and decoding self-synchronizing SDC DNA sequences. The key to our algorithms is adding an all-one vector $\1$ to make a DNA codeword complementary whenever it has the first symbol identical to the last symbol of the previous codeword, preventing confusion in Observation 2. 
	
	In the pseudo-algorithms of Algorithms \ref{alg:novel_decoding}, \ref{alg:encoding1} and \ref{alg:decoding1}, we assume that $\CC$ is a CIS code of length $2n$ over $\F_4$ having an all-one vector $\1$ as a codeword, $\cc_i$, $1\le i\le m$, are codewords in $\CC$, and $\psi$ is a map that satisfies the conditions in Theorem \ref{main_thm}. $\Sigma[i]$ denotes the $i$-th element of $\Sigma$ and $\Sigma[i .. k]$ denotes the subsequence of $\Sigma$ made from $i$th to $j$th consecutive symbols of $\Sigma$.
	
	Since we assume no delimiters, we admit that every subsequence of length $ 2n+2 $ is a possible codeword. Therefore, we modify Algorithm \ref{alg:prev_decoding} and propose Algorithm \ref{alg:novel_decoding} for detecting and correcting a single-deletion error in a subsequence of length $2n+2$.
	\begin{algorithm}[h]
		\caption{Decoding algorithm for a received vector}\label{alg:novel_decoding}
		\begin{algorithmic}[1]
			\Require a received vector $\x$ through a single-deletion channel from $\CC_{\psi}$ of length $2n+2$.		
			\Ensure the codeword $\cc$ in $\CC$ decoded from $\x$ with values of $\psi=x_{n+1}$ and $is\_del$.
			\State $L \gets$ the length of $\x$
			\If {$L=2n+1$}; $is\_del \gets true$
			\State decompose $\x=\u \oplus (x_{n+1}) \oplus \v$, where $\u$ and $\v$ are vectors of length $n$.
			\State $\phi \gets x_{n+1}$
			\If {$x_{n+1} \ne x_{n+2}$}
			\State there is a single deletion in the position of $[1..n+2]$;
			\State obtain $\cc$ generated by $\x[n+2..2n+1]$ with the information set $\{n+1,\ldots,2n\}$.
			\ElsIf {$x_{n+1} = x_{n+2}$}
			\State there is no single-deletion in $\u$
			\State obtain $\cc$ generated by $\u$ with the information set $\{1,\ldots,n\}$.
			\EndIf
			
			\ElsIf {$L=2n+2$}
			\State decompose $\x= \u \oplus (x_{n+1}, x_{n+2}) \oplus \v$, where $\u$ and $\v$ are vectors of length $n$.
			\State $\phi \gets x_{n+1}$
			\If {$x_{n+1} \ne x_{n+2}$}
			\State there is a single deletion in the position of $[1..n+2]$; $is\_del \gets true$
			\State obtain $\cc$ generated by $\x[n+2..2n+1]$ with the information set $\{n+1,\ldots,2n\}$.
			\ElsIf {$x_{n+1} = x_{n+2}$}
			\State there is no single-deletion in $\u$
			\State obtain $\cc$ generated by $\u$ with the information set $\{1,\ldots,n\}$.
			\If {$\cc[n+1..2n] = \x[n+3..2n+2]$}
			\State there is no single-deletion; $is\_del \gets false$
			\Else
			\State there is a single-deletion at $[n+3..2n]$; $is\_del \gets true$
			\EndIf
			\EndIf
			\EndIf
			
			\State return $\cc, \psi=x_{n+1}, is\_del$ 
			
		\end{algorithmic}
	\end{algorithm}

	Using Algorithm \ref{alg:novel_decoding} and Theorem \ref{main_thm}, we prove the following.
	
	\begin{theorem}\label{main_thm2}
		Let $\psi$ be a map satisfying the conditions in Theorem \ref{main_thm}, $\CC$ be a CIS code of length $2n$ over $\F_4$ having an all-one vector as codewords, and $\CC_{\psi}$ be the set of vectors $\x_{\psi}$ for all codewords $\x$ in $\CC$. Suppose that a sequence of $m$ codewords in $\CC$, $(\cc_1, \cc_2, \cdots, \cc_m)$, is encoded to a continuous sequence of codewords in $\CC_{\psi}$, $\dd_1\dd_2\cdots\dd_m$ without delimiters, per the following encoding rules:
		
		\begin{enumerate}
			\item[i)] $\dd_1 = {\cc_1}_\psi$.
			\item[ii)] For $i\ge2$, $\dd_i = {\cc_i}_\psi$ if the last symbol of $\cc_{i-1}$ is not equal to the first symbol of $\cc_{i}$.
			\item[iii)] For $i\ge2$, $\dd_i = {\cc_i}_\psi+\1$ if the last symbol of $\cc_{i-1}$ is equal to the first symbol of $\cc_{i}$, where $\1$ is the all-one vector of length $2n+2$.	
		\end{enumerate}		
		Then, the encoded sequence $\dd_1\dd_2\cdots\dd_m$ without delimiters is a self-synchronizing SDC sequence. 
		
		\begin{proof}
			We prove the theorem by induction on $m$, the number of codewords. When $m=1$, any single-deletion-error in the first codeword $\dd_1$ can be decoded using Algorithm \ref{alg:prev_decoding}. That is, a sequence consisting of a single codeword $\dd_1$ forms a self-synchronizing SDC sequence.
			Suppose, as the induction hypothesis, that for a positive integer $k$, a sequence of $k$ codewords without delimiters, $\dd_1\dd_2\cdots\dd_k$, is a self-synchronizing SDC sequence. Now consider a sequence of $k+1$ codewords, $\dd_1\dd_2\cdots\dd_k\dd_{k+1}$ without delimiters. By the induction hypothesis, the first $k$ codewords $\dd_1\dd_2\cdots\dd_k$ form a self-synchronizing SDC sequence. Thus, it suffices to prove that any single-deletion error occuring in $\dd_{k+1}$ can be corrected. 
			If the last symbol of $\dd_{k}$ is not equal to the first symbol of $\dd_{k+1}$,there is no ambiguity in determining the starting position of $\dd_{k+1}$, and thus $\dd_{k+1}$ can be decoded using Algorithm \ref{alg:novel_decoding}. 
			If, on the other hand, the last symbol of $\dd_{k}$ equals the first symbol of $\dd_{k+1}$ and if the last symbol of $\dd_{k}$ is deleted, then Algorithm \ref{alg:novel_decoding} cannot detect the single-deletion in $\dd_k$, will mistakenly interpret the first symbol of $\dd_{k+1}$ as the last symbol of $\dd_{k}$. If a single-deletion occurs within $\dd_{k+1}$, Algorithm \ref{alg:novel_decoding} would process a subsequence $\dd_{k+1}$ with two deleted symbols, resulting in a decoding failure of $\dd_{k+1}$. 
			However, encoding rule (iii) ensures that the last symbol of $\dd_{k}$ is always different from the first symbol of $\dd_{k+1}$. Therefore, no confusion arises when determining the first symbol of $\dd_{k+1}$ and $\dd_{k+1}$ can be correctly decoded using Algorithm \ref{alg:novel_decoding}. This completes the induction and the proof. 			
		\end{proof}
	\end{theorem}

	Next, we propose Algorithms \ref{alg:encoding1} and \ref{alg:decoding1}. Based on these algorithms, we propose a method for self-synchronizing SDC DNA codes in Theorem \ref{main_thm3}.

	\begin{algorithm}[h]
		\caption{Encoding algorithm for self-synchronizing SDC codes}\label{alg:encoding1}
		\begin{algorithmic}[1]
			\Require a sequence of $m$ codewords $(\cc_1, \cc_2, \cdots, \cc_r)$
			\Ensure sequence $\Sigma=\dd_1\dd_2\cdots\dd_r$ with no delimiters.
			\State $\x \gets \cc_1$
			\State convert $\x$ to $\x_\psi$.
			\State $\Sigma \gets \x_\psi$ 
			\For{$n=2, \dots, r$}
			\State $\x \gets$ the $i$-th codeword $\cc_i$
			\State convert $\x$ to $\x_\psi$.
			\State $\sigma \gets$ the first symbol of $\x_\psi$
			\If{$\sigma = \lambda$ }
			\State $\x_{\psi} \gets \x_{\psi}+\1$  
			\EndIf 
			\State concatenate $\x_{\psi}$ to $\Sigma$ without delimiter
			\State $\lambda \gets$ the last symbol of $\Sigma$.
			\EndFor
		\end{algorithmic}
	\end{algorithm}

	\begin{algorithm}[H]
		\caption{Decoding algorithm for self-synchronizing SDC codes}\label{alg:decoding1}
		\begin{algorithmic}[1]
			\Require a sequence $\Sigma=\dd_1\dd_2\cdots\dd_r$ with no delimiters and possible single-deletion errors, encoded using Algorithm \ref{alg:encoding1}. 
			\Ensure  sequence $\Lambda=(\cc_1, \cc_2, \cdots, \cc_r)$	
			\State $\Lambda \gets$ the empty sequence
			\State $m \gets$ the number of symbols in sequence $\Sigma$
			\While{$m>0$}
			\If{$m \ge 2n+2$}  
			\State $\dd \gets \Sigma[1..(2n+2)]$
			\ElsIf{$m=2n+1$}  
			\State $\dd \gets \Sigma[1..(2n+1)]$
			\Else  
			\State $\Sigma$ is undecodable; terminate.
			\EndIf 
			
			\State apply Algorithm \ref{alg:novel_decoding} on $\dd$ and obtain $\cc$ and $\phi$
			\If{$\dd$ is proved to have a single-deletion} 
			\State $\Sigma \gets \Lambda[(2n+2)..m]$
			\State $m \gets m-(2n+1)$
			\Else  
			\State $\Sigma \gets \Lambda[(2n+3)..m]$
			\State $m \gets m-(2n+2)$		 
			\EndIf 
			
			\State $a \gets \cc[n] ; b \gets \cc[n+3]$
			\If{$\psi(a,b) = \phi + 1$}  
			\State $\cc \gets \cc + \1$
			\Else  
			\State pass
			\EndIf 
			\State append $\cc$ to $\Lambda$
			\EndWhile
			\State return $\Lambda$
		\end{algorithmic}
	\end{algorithm}

	\begin{theorem}\label{main_thm3}
		Let $\psi$ be the map defined in Definition \ref{def_psi}. Assume that $\CC$ is a CIS code of length $2n$ over $\F_4$ having all-one vector as codewords, and let $\CC_{\psi}$ be the set of vectors $\x_{\psi}$ for all codewords $\x$ in $\CC$. Then, $\DD = \mu(\CC_\psi)$ is a self-synchronizing single-deletion-correcting DNA code, and its encoding and decoding algorithms can be achieved using Algorithms \ref{alg:encoding1} and \ref{alg:decoding1}.
		
		\begin{proof}
			This result is straightforward from Theorem \ref{main_thm2}. 
		\end{proof}
	\end{theorem}

	\section{Implementation on DNA codes}

	This section proposes the encoding and decoding algorithms for binary data using a self-synchronizing SDC DNA code. We implement these algorithms with Python. 
	We assume that $\CC$ is a CIS code of length $2n$ over $\F_4$ having an all-one vector, $\psi$ is the map defined in Definition \ref{def_psi}, and $n_h$ is a sufficiently large fixed positive integer such that $2^{n_h} \ge$ is the length of $bin\_data$.

	\begin{alg}[Encoding algorithm for self-synchroning SDC DNA codes]\label{alg:encoding2}
		
		\ 	
		
		{Input:}  		$bin\_data$:= binary data, such as. txt file, image files, etc. \\ 
		{Output:} 		$\Sigma$ := encoded DNA sequence of $bin\_data$\\  
		
		\begin{enumerate}[Step 1.] 
			\item{} [Convert $bin\_data$ to $bin\_seq$] \\
			Add header as metadata for $bin\_data$ and zero-padding so that the length of the converted binary sequence, $bin\_seq$, is a multiple of $2n$.
			\begin{itemize}
				\item[-] (Length check) let $\ell$ be the length of $bin\_data$  
				\item[-] (Header) convert $\ell$ to binary number $\ell_{2}$ of $n_h$ digits.  
				\item[-] (Zero-padding) let $\ell_{p}$ be all 0 sequence of length $2n-(\ell+n_h)$ modulo $2n$. 
				\item[-] Let $bin\_seq$ be the concatenated sequence of $\ell_{2}$, $\ell_{p}$, and $bin\_data$, in this order. 
			\end{itemize}
			
			\item{}[Convert $bin\_seq$ to a pre-$dna\_seq$]	\\    
			Convert every two each symbols in $bin\_data$ to a DNA symbol, $A,C,G$ and $T$ under the map $\delta$ to obtain pre-$dna\_seq$: 
			$${\text pre-}dna\_seq= \delta(bin\_seq).$$
			The length of pre-$dna\_seq$ is a multiple of $n$.

			\item{}[Divide pre-$dna\_seq$]\\	
			Divide every $n$ symbol of pre-$dna\_seq$ and let $r$ be the number of divisions. Then, pre-$dna\_seq$ is in the form 
			$$ (\mm_1, \mm_2, \cdots, \mm_r),$$ where $\mm_i$ is a sequence of $n$ DNA symbols for $1\le i \le r$.
			
			\item{}[Encode each block of pre-$dna\_seq$]\\	
			Encode each $\mm_i$ of pre-$dna\_seq$ to DNA codeword of $\CC$ as follows.
			\begin{itemize}
				\item[-] Convert each $\mm_i$ to a $2 \times 2n $ matrix $f(\mm_i)$ over $\F_2$ by map $f$. 
				\item[-] Encode each $f(\mm_i)$ to a codeword $f(\cc_i)$ by multiplying the generator matrix of $\CC$ under map $\tau$.
				\item[-] Apply map $f^{-1}$ to $f(\cc_i)$ and obtain the DNA codeword $\cc_i$ of $\CC$.
			\end{itemize}
			Thus, we obtain the sequence of DNA codewords $(\cc_1, \cc_2, \cdots, \cc_r)$
			
			\item{}[Encoding to self-synchronizing SDC sequence]\\
			Apply Algorithm \ref{alg:encoding1} to the sequence of DNA codewords $(\cc_1, \cc_2, \cdots, \cc_r)$.
			Finally, we obtain the DNA sequence $\Sigma=\dd_1\dd_2\cdots\dd_r$ with no delimiters.

		\end{enumerate}
	\end{alg}

	\begin{alg}[Decoding algorithm]\label{alg:decoding2}
		
		\ 	
		
		{Input:}  		 DNA sequence $\Sigma=\dd_1\dd_2\cdots\dd_r$ with possible single-deletion errors and no delimiters. \\ 
		{Output:} 		$bin\_data$ := the original binary data.\\

		\begin{enumerate}[Step 1.] 
			\item{} [Decode $\Sigma$]\\
			Apply Algorithm \ref{alg:decoding1} to the DNA sequence $\Sigma=\dd_1\dd_2\cdots\dd_r$ with no delimiters. Then we obtain the sequence of DNA codewords $\Lambda = (\cc_1, \cc_2, \cdots, \cc_r)$
			
			\item{} [Obtain pre-$dna\_seq$] \\
			For each DNA codeword $\cc_i$ in $\Lambda$, let $\mm_i$ be $\cc_i[1..n]$. Concatenating $\mm_i$ for all $i$ gives the  pre-$dna\_seq$.

			\item{} [Convert pre-$dna\_seq$ to $bin\_seq$]\\
			Obtain $bin\_seq$ using the inverse map of $\delta$, that is, $$bin\_seq = \delta^{-1}({\text pre-}dna\_seq).$$

			\item{} [Obtain original binary data]\\
			The binary number made by the first $n_h$ digits of $bin\_seq$ is the number of the length of the original binary data. Take that amount of digits of $bin\_seq$, counting from the end of  $bin\_seq$, to obtain the original binary data $bin\_data$.

		\end{enumerate}
	\end{alg}

	We provide the following examples, which illustrate Algorithms \ref{alg:encoding2} and \ref{alg:decoding2}. In the following examples, we use a reversible CIS [6,3,3]-code in the encoding and decoding process. We exploited the reversibility of codewords when decoding single-error codewords. For details on the reversible code, please refer \cite{choi2020, kim2020}.
	
	\begin{example}\label{ex1}
		Let $\CC$ be a reversible self-dual code of length $6$ over $\F_4$ with generator matrix
		$$G=\begin{pmatrix}
			1&0&0&\omega&1&\omega\\0&1&0&\omega^2&\omega^2&1 \\ 0 &0&1&0 &\omega^2&\omega \\
		\end{pmatrix}
		,$$
		and set the header length to $n_h= 6$. It is easy to verify that $\CC$ is an all-one vector as a codeword.
		The conversion of the generator matrix over DNA symbols is 
		$$\mu^{-1}(G)=\begin{pmatrix}
			T&A&A&C&T&C\\A&T&A&G&G&T \\ A &A&T&A &G&C \\
		\end{pmatrix}
		,$$
		and the conversion of the generator matrix over GF(2) is 
		$$\tau(G)=
		\begin{pmatrix}
			100000011001\\ 
			010000110111\\ 
			001000111110\\ 
			000100101001\\ 
			000010001101\\ 
			000001001011\\
		\end{pmatrix}
		.$$
		
		Assume that we have binary data of 19 bits: 
		$$bin\_data=
		10101
		00101 	
		01010
		0111.$$
		
		\begin{enumerate}[Step 1.]
			
			\item{} [Convert $bin\_data$ to $bin\_seq$] \\
			Since we set the header length $n_h= 6$, the header becomes $010011$ as the six-digit binary representation of length $19$. Thus, we need five $0$'s for zero-padding so that the concatenated binary sequence has a length of 30, a multiple of $6$:
			$$bin\_seq=01 00 11 /00 00 0/1
			01 01 00 10
			10 10 10 01 11.$$
			Note that there are $m=5$ blocks in this sequence. 
			
			\item{}[Convert $bin\_seq$ to pre-$dna\_seq$]	\\    
			The concatenated data is converted to the pre-$dna\_seq$ of length $30/2=15$: 
			$$\text{pre-}dna\_seq = \delta(bin\_seq)\\=GATAA GGGAC CCCGT.$$ 
			
			\item{}[Divide pre-$dna\_seq$]\\	
			The sequence $GATAAGGGACCCCGT$ is divided into five subsequences of 3 DNA symbols:
			$$GAT,AAG,GGA,CCC,CGT.$$
			
			\item{}[Encode each block of pre-$dna\_seq$]\\
			Each subsequence of pre-$dna\_seq$ is converted to a $2 \times 6$ matrix using map $f$: 
			{\small $$\begin{pmatrix}
					11 00 10 \\
					10 00 01  
				\end{pmatrix}, 
				\begin{pmatrix}
					00 00  11\\     
					00 00  10 
				\end{pmatrix}, 
				\begin{pmatrix}
					11 11  00\\ 
					10 10 00 
				\end{pmatrix},
				\begin{pmatrix}
					01  01 01\\  
					11  11  11    
				\end{pmatrix},
				\begin{pmatrix}
					01 11 10 \\
					11  10 01 
				\end{pmatrix}.$$}
			By multiplying the generator matrix $\tau(G)$ with each block, we obtain five codewords of $2 \times 12$ matrices. For example, the first block is encoded as 
			$$ \begin{pmatrix}
				11 00 10  \\
				10 00 01 \\
			\end{pmatrix}	\begin{pmatrix}
				100000011001\\ 
				010000110111\\ 
				001000111110\\ 
				000100101001\\ 
				000010001101\\ 
				000001001011\\
			\end{pmatrix}=\begin{pmatrix}
				11 00 10 100011 \\
				10 00 01 010010\\
			\end{pmatrix}.$$
			Applying $f^{-1}$ gives 
			$$f^{-1}\begin{pmatrix}
				11 00 10 100011 \\
				10 00 01 010010\\
			\end{pmatrix}=GATTAG.$$
			
			Repeating this process for all the codewords and concatenating them, we obtain the sequence of DNA codewords
		$$
				GATTAG, AAGACT, GGAGTC, CCCCCC, CGTTGC.
		$$	
			\item{}[Encoding to self-synchronizing SDC sequence]\\
			Applying Algorithm \ref{alg:decoding2}, the sequence of DNA codewords is converted into 
			\begin{multline*}GATCCTAG, AAGTTACT, GGATTGTC,\\ GGGTTGGG, CGTCCTGC.\end{multline*}	
			The fourth block $CCCCCC$ is converted into $CCCAACCC$ at first; however, the first symbol $C$ is identical to the last symbol of the previous codeword $GGATTGTC$. Thus, we take $GGGTTGGG$, the complement of $CCCAACCC$.
			Therefore, as the self-synchronizing SDC DNA sequence, we obtain 
			\begin{multline*}\Sigma=GATCCTAG AAGTTACT GGA\\TTGTC GGGTTGGG CGTCCTGC.\end{multline*}	
		\end{enumerate}
	\end{example}

	\begin{example}\label{ex2}
		Assume that the sender sends the original DNA sequence from Example \ref{ex1}: 
		$$\Sigma=GATCCTAG AAGTTACT GGATTGTC GGGTTGGG CGTCCTGC,$$
		and assume that we have the information of the code $\CC$ and $n_h= 6$. Suppose that four single-deletion errors occur during transmission as follows:
		$$
			\Sigma = GATCCTAG A\bcancel{A}GTTACT GGA\bcancel{T}TGTC GGGTTG\bcancel{G}G C\bcancel{G}TCCTGC.$$
		Thus, we receive the following DNA sequence:  
		$$\Sigma'= GATCCTAG AGTTACT GGATGTC GGGTTGG CTCCTGC$$
		\begin{enumerate} [Step 1.]
			
			\item{} [Decode $\Sigma'$] \\
			Apply Algorithm \ref{alg:decoding1} to the DNA sequence $\Sigma'$.
			
			\begin{itemize}
				\item[-] Set $\Lambda=()$, empty sequence and $m=36$, the number of symbols in $\Sigma'$
				\item[-] Since $m =36 > 0$, the first iteration begins. 
				\begin{itemize}
					\item[(1)] $m>8$; thus, we set $\dd = GATCCTAG$.
					\item[(2)] Apply Algorithm \ref{alg:novel_decoding} on $\dd$:
					\begin{itemize}
						\item[-] $\dd=GATCCTAG$ is decomposed into $GAT,CC,TAG$ and set $\phi=C$. 
						\item[-] Since $x_{n+1}=C=x_{n+2}$, there is no deletion in position $\dd[1..3]=GAT$.
						\item[-] $f(GAT) =\begin{pmatrix}
							11 00 10  \\
							10 00 01 \\
						\end{pmatrix},$ and multiplying $\tau(G)$, we obtain \\ $f(\cc) = \begin{pmatrix}
							11 00 10 100011 \\
							10 00 01 010010\\
						\end{pmatrix};$ therefore, $\cc=GATTAG.$
					\end{itemize}
					\item[(3)] Since $\cc[4..6] = \dd[6..8]$, we conclude that there is no deletion. Thus, $m=m-8=28$ and $\Sigma'$ becomes$$\Sigma'=AGTTACT GGATGTC GGGTTGG CTCCTGC.$$
					\item[(4)] Since $\psi(a,b)=\psi(T,T)=C=\phi$, we append $\cc$ to $\Lambda$, that is, $\Lambda=(GATTAG)$ and the first iteration ends.
				\end{itemize}
				\item[-] Since $m =28> 0$, the second iteration begins. 
				\begin{itemize}
					\item[(1)] $m>8$; thus, we set $\dd = AGTTACTG $.
					\item[(2)] Apply Algorithm \ref{alg:novel_decoding} on $\dd$:
					\begin{itemize}
						\item[-] $\dd=AGTTACTG$ is decomposed into $AGT,TA,CTG$, and set $\phi=T$.
						\item[-] Since $x_{n+1}=T \ne A =x_{n+2}$, there is a single deletion in position [1..5].
						\item[-] We take $f(\dd[5..7])=f(ACT)$ with information set [4..6]. To obtain $\cc$, we use $f(TCA^r)=f(TCA)=\begin{pmatrix}
							10 01 00  \\
							01 11 00 \\
						\end{pmatrix}.$ By multiplying $\tau(G)$, we obtain
						
					$\begin{pmatrix}
								10 01 00  \\
								01 11 00 \\
							\end{pmatrix}	\begin{pmatrix}
								100000011001\\ 
								010000110111\\ 
								001000111110\\ 
								000100101001\\ 
								000010001101\\ 
								000001001011\\
							\end{pmatrix}=\begin{pmatrix}
								10 01 00 11 00 00\\
								01 11 00 10 00 00\\
							\end{pmatrix}.$
						
						 Thus, we see that $TCAGAA$ is a codeword of $\CC$, and the reversibility of $\CC$ ensure that $AAGACT$ is also a codeword in $\CC$ having $ACT$ in information set [4..6]. Therefore, we conclude that $\cc=AAGACT$.
						
					\end{itemize}
					\item[(3)] Since $\dd$ has a single-deletion, $m=m-7=21$ and $\Sigma'$ becomes $$\Sigma'=GGATGTC GGGTTGG CTCCTGC,$$ returning the last symbol $G$ of $\dd$.
					\item[(4)] Since $\psi(a,b)=\psi(G,A)=T=\phi$, we append $\cc$ to $\Lambda$, and the second iteration ends.
				\end{itemize}
				
				\item[-] Since $m =21> 0$, the third iteration begins. 
				\begin{itemize}
					\item[(1)] $m>8$; thus, we set $\dd = GGATGTCG $.
					\item[(2)] Apply Algorithm \ref{alg:novel_decoding} on $\dd$:
					\begin{itemize}
						\item[-] $\dd=GGATGTCG$ is decomposed into $GGA TG TCG$, and set $\phi=T$.
						\item[-] Since $x_{n+1}=T \ne G =x_{n+2}$, there is a single deletion in position [1..5].
						\item[-] We take $f(\dd[5..7])=f(TCG)$ with information set [4..6]. 
						\item[-] We take a similar process as the second iteration, and obtain $\cc = GGAGTC$.
					\end{itemize}
					\item[(3)] Since $\dd$ has a single-deletion, $m=m-7=14$ and $\Sigma'$ becomes $$\Sigma'=GGGTTGG CTCCTGC,$$ returning the last symbol $G$ of $\dd$.
					\item[(4)] Since $\psi(a,b)=\psi(A,G)=T=\phi$, we append $\cc$ to $\Lambda$, and the third iteration ends.
				\end{itemize}
				
				\item[-] Since $m =14> 0$, the fourth iteration begins. 
				\begin{itemize}
					\item[(1)] $m>8$; thus, we set $\dd = GGGTTGG C $.
					\item[(2)] Apply Algorithm \ref{alg:novel_decoding} on $\dd$:
					\begin{itemize}
						\item[-] $\dd=GGGTTGG C$ is decomposed into $GGG, TT, GGC$, and set $\phi=T$.
						\item[-] Since $x_{n+1}=T =x_{n+2}$, there is no single deletion in position [1..5].
						\item[-] We take $f(\dd[1..3])=f(GGG)$ with information set [1..3]. 
						\item[-] We take a similar process as the first iteration, and obtain $\cc = GGGGGG$.
						
					\end{itemize}
					\item[(3)] Since $\cc[4..6] = GGG \ne  GGC = \dd[6..8]$, $\dd$ has a single-deletion. Thus $m=m-7=7$ and $\Sigma'$ becomes $$\Sigma'=CTCCTGC,$$ returning the last symbol $C$ of $\dd$.
					\item[(4)] Since $\psi(G,G)=A \ne T=\phi$, we append $\cc+\1=CCCCCC$ to $\Lambda$, and the fourth iteration ends.
				\end{itemize}
				
				\item[-] Since $m = 7 > 0$, the fifth iteration begins. 
				\begin{itemize}
					\item[(1)] $m = 7 = 2n+1$; thus, we set $\dd = CTCCTGC$.
					\item[(2)] Apply Algorithm \ref{alg:novel_decoding} on $\dd$:
					\begin{itemize}
						\item[-] $\dd=CTCCTGC$ is decomposed into $CTC,C,TGC$, and set $\phi=C$.
						\item[-] Since $x_{n+1}=C \ne T =x_{n+2}$, there is a single deletion in position [1..5].
						\item[-] We take $f(\dd[5..7])=f(TGC)$ with information set [4..6].
						\item[-] We take a similar process as the second iteration, and obtain $\cc = CGTTGC$.
						
					\end{itemize}
					\item[(3)] Since $\dd$ has a single-deletion, $m=m-7=0$ and $\Sigma'$ becomes empty.
					\item[(4)] Since $\psi(a,b)=\psi(T,T)=C=\phi$, we append $\cc$ to $\Lambda$, and the fifth iteration ends.
				\end{itemize}
				\item[-] Since $m = 0$, it is terminated.
				\item[-] It returns $\Lambda=(GATTAG, AAGACT, GGAGTC, CCCCCC, CGTTGC).$
				
			\end{itemize}
			
			\item{} [Obtain pre-$dna\_seq$] \\
			From $\Lambda=(GATTAG, AAGACT, GGAGTC,CCCCCC, CGTTGC),$ we obtain the pre-$dna\_seq$: 
			$$GATAAGGGACCCCGT$$
			
			\item{} [Convert pre-$dna\_seq$ to $bin\_seq$]\\
			The map $\delta^{-1} (GATAAGGGACCCCGT)$ gives $$bin\_seq=01 00 11 00 00 01
			01 01 00 10
			10 10 10 01 11.$$

			\item{} [Obtain the original binary data]\\
			Since $n_h=6$, the length of the original binary data is binary $010011$, equivalently, 19. Therefore we take 19 digits of $bin\_seq$ from the end, and the original binary data $bin\_data$ is $1
			01 01 00 10
			10 10 10 01 11.$
		\end{enumerate}
	\end{example}

%
%
%

	\section{Concluding Remarks}

	This study introduces a novel approach for correcting single-deletion errors in continuous transmissions without delimiters with a self-synchronizing  capability. Whereas traditional error-correcting codes concentrate only on substitution errors, applications of coding theory in bioinformatics encompass a wider range of errors, including deletions, insertions, and erasures, known as synchronization errors. The historical context of the self-synchronization problem in biology, particularly in DNA sequences, has been explored since the discovery of DNA. We point out that Francis Crick's early proposal of ``codes without commas'' for DNA sequences, though proven biologically incorrect, inspired mathematicians interested in synchronization problems.
	
	While this work provides a theoretical foundation for the construction of single-deletion-correcting DNA codes using CIS codes, we acknowledge that our focus has been on mathematical formulation and analysis rather than immediate practical implementation. As a result, certain challenges remain regarding the direct application of our approach to real-world DNA storage and DNA computing systems. Addressing these practical aspects—including experimental validation and adaptation to the constraints of DNA synthesis, sequencing, and channel noise—will be an important direction for future research. We hope that our results will serve as a stepping stone for both theoretical advances and future applications in the field.

	We hope to find applications of self-synchronizing single-deletion correcting codes in various research fields, such as DNA computing and computer designs, DNA data-storage devices, and synthetic DNA sequence designs. In future work, we will explore the applications of DNA coding theory and continue to study self-synchronizing codes with multi-deletion or insertion-correcting capabilities.

	\section*{Acknowledgment}
	This work is supported by the
		the National Research Foundation of Korea (NRF) grant funded by the
		Korea government (2019R1I1A1A01057755, 2022R1C1C2011689).


\begin{thebibliography}{00}

		
		
		
		\bibitem{ghaffar2010correcting} K.A. Abdel-Ghaffar, H.C. Ferreira, and L. Cheng, “Correcting deletions using linear and cyclic codes,'' \textit{IEEE Trans. Inform. Theory}, vol. 56, no. 10, pp. 5223--5234, 2010. 
		
		\bibitem{Magma} J. Cannon, C. Playoust, ``An Introduction to Magma.'' University of Sydney, Sydney, Australia, 1994.
		
		\bibitem{choi2020} W.-H. Choi, H.J. Kim, and Y. Lee, ``Construction of single-deletion-correcting DNA codes using CIS codes,'' \textit{Des. Codes Cryptogr.}, vol. 88, pp. 2581--2596, 2020. 
		
		\bibitem{CGKS} C. Carlet, P. Gaborit, J-L. Kim, and P. Sol$\acute{\text{e}}$, ``A new class of codes for Boolean masking of cryptographic computations,'' \textit{IEEE Trans. Inform. Theory}, vol. 58, pp. 6000--6011, 2012.
		
		\bibitem{crick} F.H.C. Crick, J.S. Griffith, and L. E. Orgel, ``Codes without commas,'' \textit{Proc. Natl. Acad. Sci. U.S.A.}, vol. 43, no. 5, pp. 416--421, 1957.
		
		\bibitem{gaborit2005} P. Gaborit and O.D. King, ``Linear constructions for DNA codes,'' \textit{Theoret. Comput. Sci.}, vol. 334(1-3), pp. 99--113, 2005.
		
		\bibitem{Golomb} S. W. Golomb, G. Basil, and R.W. Lloyd, ``Comma-free codes,'' \textit{Can. J. Math}, vol. 10, pp. 202--209, 1958.
		
	
		\bibitem{Hanna2023} 
		S.K. Hanna, ``Effective IDS Error Correction Algorithms for DNA Storage Channels With Multiple Output Sequences,'' \textit{IEEE Trans. Inf. Theory}, vol. 69, no. 9, pp. 5687--5700, 2023.

		\bibitem{Haeupler2021}
		B. Haeupler, A. Shahrasbi, ``Synchronization Strings and Codes for Insertions and Deletions - A Survey,'' \textit{IEEE Trans. Inf. Theory}, vol. 67, no. 6, pp. 3190--3206, 2021.

		
		\bibitem{HuffmanPless} W.C. Huffman, V. Pless, \textit{Fundamentals of Error-Correcting Codes}, Cambridge University Press, Cambridge, 2003.
		
		\bibitem{kim2020} H.J. Kim, W.-H. Choi, and Y. Lee, ``Construction of reversible self-dual codes,'' \textit{Finite Fields Appl.}, vol. 67, pp. 101714, 2020. 
		
		\bibitem{kim2021} H.J. Kim, W.-H. Choi, and Y. Lee, ``Designing DNA codes from reversible self-dual codes over $GF(4)$,'' \textit{Discrete Math.}, vol. 344, no. 1, pp. 112159, 2021.
		
		
		\bibitem{kim2016cis} H.J. Kim and Y. Lee, ``Complementary information set codes over $GF(p)$,'' \textit{Des. Codes Cryptogr.}, vol. 81, pp. 541--555, 2016.
		
		\bibitem{king2003} O.D. King, ``Bounds for DNA codes with constant GC-content,'' \textit{Electron. J. Combin.}, vol. 10, R33, 2003.
		
		
		\bibitem{levenshtein1966binary} V.I. Levenshtein, ``Binary codes capable of correcting deletions, insertions, and reversals,'' \textit{In Soviet physics doklady}, vol. 10, no. 8, pp. 707--710, 1966.
		
		
		\bibitem{levy} J. Levy, ``Self-synchronizing codes derived from binary cyclic codes,'' \textit{IEEE Trans. Inf. Theory}, vol. 12, no. 3, pp. 286--290, 1966.
		
		
		\bibitem{MS} F.J. MacWilliams, N.J.A. Sloane, \textit{The theory of error-correcting codes}, North-Holland, Amsterdam, 1977.
		

		\bibitem{Massey1972}
		J.L. Massey, ``Optimum Frame Synchronization,'' \textit{IEEE Trans. Commun.}, vol. 20, no. 2, pp. 115--119, 1972.

		
		\bibitem{Sal}
		F. Sala, R. Gabrys, C. Schoeny, and L. Dolecek, ``Exact Reconstruction From Insertions in Synchronization Codes,'' \textit{IEEE Trans. Inf. Theory}, vol. 63, no. 4, pp. 2428--2445, 2017. 
			
		\bibitem{Python} Van Rossum, G., \& Drake, F. L. (2009). Python 3 Reference Manual. Scotts Valley, CA: CreateSpace.
			
		\bibitem{Yan2023}
		Z. Yan, C. Liang, and H. Wu, ``A Segmented-Edit Error-Correcting Code With Re-Synchronization Function for DNA-Based Storage Systems,'' \textit{IEEE Trans. Emerg. Top. Comput.}, vol. 11, no. 3, pp. 605-618, 2023.



	\end{thebibliography}
\end{document}